\documentclass{article}
\usepackage[T1]{fontenc}
\usepackage[utf8]{inputenc}
\usepackage[english]{babel}

\usepackage{amsmath,amssymb, amsthm}
\usepackage{thm-restate}
\usepackage{enumitem}

\usepackage{hyperref}
\usepackage{cleveref}

\newtheorem{theorem}{Theorem}

\usepackage{color}

\newcommand{\GHD}{\mathsf{GHD}}

\newcommand{\E}{\mathrm{E}}
\newcommand{\D}{\mathrm{D}}

\newcommand{\diam}{\mathrm{diam}}

\newtheorem{lemma}{Lemma}
\newtheorem{proposition}{Proposition}
\newtheorem{definition}{Definition}
\title{Some Bounds on Communication Complexity of Gap Hamming Distance}

\author{Alexander Kozachinskiy\footnotemark[1]\\[2mm]
\footnotemark[1]~{Lomonosov Moscow State University, kozlach@mail.ru}}

\begin{document}
\maketitle
\begin{abstract}
In this paper we obtain some bounds on communication complexity of Gap Hamming Distance problem ($\GHD^n_{L, U}$): Alice and Bob are given binary string of length $n$ and they are guaranteed that Hamming distance between their inputs is either $\le L$ or $\ge U$ for some $L < U$. They have to output 0, if the first inequality holds, and 1, if the second inequality holds. 

In this paper we study the communication 
complexity of $\GHD^n_{L, U}$ for probabilistic
protocols with one-sided error and for deterministic
protocols. 
Our first result is a protocol  
which communicates 
$O\left(\left(\frac{s}{U}\right)^\frac{1}{3} \cdot n\log n\right)$ bits
and has  one-sided error probability $e^{-s}$ provided 
$s \ge \frac{(L + \frac{10}{n})^3}{U^2}$.

Our second result is about 
deterministic communication complexity of $\GHD^n_{0,\, t}$.
Surprisingly, it can be computed with logarithmic precision:
$$\D(\GHD^n_{0,\, t}) =  n - \log_2 V_2\left(n, \left\lfloor\frac{t}{2}\right\rfloor\right) + O(\log n),$$
where $V_2(n, r)$ denotes the size of Hamming ball of radius $r$.
 As an application of this result for every $c < 2$ we prove a $\Omega\left(\frac{n(2 - c)^2}{p}\right)$ lower bound on the space complexity of any $c$-approximate deterministic $p$-pass streaming algorithm for computing the number of distinct elements in a data stream of length $n$ with tokens drawn from the universe $U = \{1, 2, \ldots, n\}$. Previously that lower bound 
was known for $c < \frac{3}{2}$ and for $c < 2$ but with larger $|U|$.

\end{abstract}

\section{Introduction}
\subsection{Gap Hamming Distance Problem}
%In this work we focus on the communication complexity of the Gap Hamming Distance problem. 
Given two strings $x = x_1\ldots x_n \in \{0, 1\}^n$, $y = y_1\ldots y_n \in \{0, 1\}^n$, Hamming distance between $x$ and $y$ is defined as the 
number of positions, where $x$ and $y$ differ:
$$H(x, y) = \left|\{i \in \{1, \ldots, n\}\,|\, x_i \neq y_i\}\right|.$$

Let $L<U\le n$ be integer numbers. 
In this paper we consider the following communication problem
$\GHD^n_{L, U}$, called the Gap Hamming Distance problem:
\begin{definition}
\label{GHD}
Let Alice receive an $n$-bit  string $x$ and Bob an  $n$-bits string $y$ 
such that either $H(x, y) \le L$, or $H(x, y) \ge U$. They have to output 0, if the first inequality holds, and 1, if the second inequality holds. If the promise is not fulfilled, they may output anything.
\end{definition}

The Gap Hamming Distance problem is motivated by the problem of approximating the number
of distinct elements in a data stream (see \cite{indyk2003tight}, \cite{brody2010better}). There is the following simple and relatively efficient protocol with shared randomness
to solve $\GHD^n_{L, U}$.
Alice and Bob pick $i \in \{1, \ldots, n\}$ uniformly at random (using shared randomness) 
and check, whether $x_i = y_i$ or not. They repeat it  many times
and then perform some kind of a majority vote: if in more than $\frac{L + U}{2n}$ fraction of trials it happened 
that $x_i = y_i$ they output 0, and they output 1 otherwise.
It can be shown, that $O\left(\frac{snU}{(U - L)^2}\right)$
number of times 
is sufficient to make error probability less than $e^{-s}$. Hence
\begin{equation}
\label{sampling_protocol}
R_{e^{-s}}\left(\GHD^n_{L, U}\right) = O\left(\frac{snU}{(U - L)^2}\right).
\end{equation}
Here $R_\varepsilon(f)$ denotes randomized public-coin communication complexity of $f$ with error probability $\varepsilon$.

Previously the Gap Hamming Distance problem was studied in the symmetric case: $L = \frac{n}{2} - \gamma, U = \frac{n}{2} + \gamma$. 
Let $\GHD^n_\gamma$ stand for $\GHD^n_{L,U}$
for these specific values of $L$ and $U$. 
In this notation the bound \eqref{sampling_protocol} becomes $O\left(\frac{n^2}{\gamma^2}\right)$ (for a constant error, say, $\frac{1}{3}$). It turns out that this bound is tight:

\begin{theorem}[\cite{chakrabarti2012optimal}]$R_\frac{1}{3}(\GHD^n_\gamma) = \Theta\left(\min\left\{\frac{n^2}{\gamma^2}, n \right\}\right).$
\end{theorem}

The most difficult case is $\gamma = c\sqrt{n}$, where $c$ is a constant,
in which case the lower  bound becomes
$R_\frac{1}{3}(\GHD^n_{c\sqrt{n}}) = \Omega(n)$.
There are several proofs of this bound \cite{chakrabarti2012optimal}, \cite{vidick2012concentration}, \cite{sherstov2012communication}. As noted in \cite{chakrabarti2012optimal}, for other values of $\gamma$ the bound 
can be proved via the following reduction: 
$$R_\frac{1}{3}(\GHD^{n/k}_{\gamma/k}) \le R_\frac{1}{3}(\GHD^n_\gamma)$$
for $k > 1$. Setting $k = \Theta\left(\frac{\gamma^2}{n}\right)$
in this inequality, we can reduce Theorem~1 to its special case.

To the best of our knowledge, 
$\GHD$ has not been studied for $L + U \neq n$ (except simple inequality (\ref{sampling_protocol})). This paper establishes
new bounds on communication complexity of $\GHD$ with different parameters and in different settings.

\subsection{Our Results}
In section \ref{section_3} we provide the following upper bound on randomized communication complexity of $\GHD$ with one-sided error. Before claim it, let us fix our notations. In this paper $R^i_\varepsilon(f)$ stands for the minimal possible depth of the communication protocol  with shared randomness, which never errs on inputs from $f^{-1}(i)$ and which errs with probability at most $\varepsilon$ on inputs from $f^{-1}(1 - i)$ (here $f$ is partial Boolean function).

\begin{theorem}
\label{upper_bound}
If $s \ge \frac{\left(L + \frac{10}{n}\right)^3}{U^2}$, then
$$R_{e^{-s}}^0 (\GHD^n_{L,\, U})= O\left( \left(\frac{s}{U}\right)^\frac{1}{3}\cdot n\log n \right).$$
\end{theorem}

Let us compare this bound with the upper bound \eqref{sampling_protocol}
for protocols with  two-sided error. For simplicity 
assume that $L < \frac{U}{2}$.% and $s<U$. 
Then \eqref{sampling_protocol} becomes
$$
R_{e^{-s}}(\GHD^n_{L,\, U}) = O\left(\frac{s}{U} \cdot n \right).
$$
Instead of $\frac{s}{U}$, theorem \ref{upper_bound} has 
$\left(\frac{s}{U}\right)^\frac{1}{3}$, 
which is bigger than  $\frac{s}{U}$ when $s < U$. 
As  $s$ tends to $U$, both 
$\frac{s}{U}$ and $ \left(\frac{s}{U}\right)^\frac{1}{3}$ tend to 1 
and both bounds become trivial.

In section \ref{section_4} we study the deterministic communication complexity of $\GHD^n_{0,\, t}$. Namely, we prove the following theorem

\begin{theorem}
\label{det}
$$\D(\GHD^n_{0,\, t}) =  n - \log_2 V_2\left(n, \left\lfloor\frac{t}{2}\right\rfloor\right) + O(\log n)$$
where $V_2(n, r)$ denotes the size of Hamming ball of radius $r$.
\end{theorem}

We use this result to prove the following lower bound on space complexity of approximating the number of distinct elements in a data stream:
\begin{theorem}
\label{streaming_lower_bound}
Assume that $1 < c < 2$ and $A$ is a $p$-pass deterministic 
streaming algorithm for estimating $F_0$, 
the number of distinct elements in a given data stream of size $2n$ 
with tokens drawn from the universe $U =\{1, 2, \ldots, 2n\}$. 
If  $A$ outputs a number $E$ such that $F_0 \le E < cF_0$, 
then $A$ must use linear space, 
namely $\Omega\left(\frac{n(2 - c)^2}{p}\right)$.
\end{theorem}

Previously such a bound 
was known in the case when the size of the 
universe is constant-time larger than the size of the 
data stream. In the case when the size of the 
universe and the size of the data stream coincide
the bound was known only for $c < \frac{3}{2}$.
%\begin{theorem}
%\label{relationship}
%$$\N^0(\GHD^n_{0,\, n - k}) \le \D(\GHD^n_{0, \,n - k}) \le \N^0(\GHD^n_{0,\, %n - k}) + 1.$$
%\end{theorem}

%This allows us to establish the connection between $\GHD$ and error-correcting codes.

%\begin{theorem} 
%\label{relationship_with_code}
%The following inequality holds
%$$\lceil\log_2 A(n, d)\rceil \ge \D(\GHD^n_{0, 2d - 1}) - 1, $$
%where $A(n,d)$
%denotes the maximum possible size of a 2-ary code with length $n$ and minimum %Hamming weight $d$.
%\end{theorem}

\section{Preliminaries}
\subsection{Communication Complexity}
%\subsection{Randomized Communication Complextity}

Let $f:\mathcal{X}\times\mathcal{Y}\to\{0, 1\}$ be a Boolean function and 
$R$ an arbitrary random variable  whose support is $\mathcal{R}$.

\begin{definition} 
\label{randomized_protocol}
A randomized 
(public-coin) 
communication protocol is a rooted binary tree, in which each non-leaf vertex is associated either with Alice or with Bob and each leaf is labeled by $0$ or $1$. For each non-leaf vertex $v$, associated with Alice, there is a function $f_v:\mathcal{X}\times\mathcal{R}\to\{0, 1\}$ and for each non-leaf vertex $u$, associated with Bob, there is a function $g_u:\mathcal{Y}\times\mathcal{R}\to\{0, 1\}$. For each non-leaf vertex one of its out-going edges is labeled by $0$ and other one is labeled by $1$.
\end{definition}

\begin{definition}
Communication complexity of a protocol $\pi$, denoted by $CC(\pi)$, is defined as the depth of the corresponding binary tree.
\end{definition}

A computation according to a protocol runs as follows. 
Alice is given $x \in \mathcal{X}$, Bob is given $y \in \mathcal{Y}$. 
They start at the root of tree. 
If they are in a non-leaf vertex $v$, associated with Alice, 
Alice sends $f_v(x, R)$ to Bob and they 
move to the son of $v$ by the edge labeled by $f_v(x, R)$. 
If they are in a non-leaf vertex, associated with Bob, they act in a 
similar same way, however this time it is Bob who sends a bit to
Alice. When they reach a leaf, they output the bit which labels this leaf.

%A communication protocol is called \emph{one-way}, 
%if Alice and Bob communicate in the following way: 
%Alice sends the first message to Bob, then 
%Bob sends one bit to Alice and then they output that bit.
 
\begin{definition}
We say that a randomized protocol  computes $f$ with error probability 
$\varepsilon$, if for every pair of inputs 
$(x, y) \in \mathcal{X}\times\mathcal{Y}$ with
probability at least $1 - \varepsilon$ 
that protocol outputs $f(x, y)$.
Randomized communication complexity of $f$ is defined as
$$
R_{\varepsilon}(f) = \min\limits_{\pi}CC(\pi),
$$
where minimum is over all protocols that  compute $f$ with error probability 
$\varepsilon$.
\end{definition}

%There are many variations of this definition. 
%If we consider only one-way protocols, then the corresponding notion is 
%called ``randomized one-way communication complexity'' and is
%denoted by $R_{\varepsilon}^{\mathsf{one-way}}(f)$. 
If for $i \in \{0, 1\}$ we require that the protocol never errs on 
inputs from $f^{-1}(i)$, then the corresponding notion is called ``randomized one-sided error communication complexity'' and is
denoted by $R_{\varepsilon}^{i}(f)$. 

If $f$ is a partial function, then, in the definition of computation 
with error we consider only inputs from the domain of $f$. 
%Obviously,
%\begin{equation}\label{c2}
%R^{i, \mathsf{one-way}}_{\varepsilon}(f)\ge 
%R^{i}_{\varepsilon}(f) \ge R_{\varepsilon}(f),\qquad 
%R^{i, \mathsf{one-way}}_{\varepsilon}(f)\ge R^{ %\mathsf{one-way}}_{\varepsilon}(f) \ge R_{\varepsilon}(f).
%\end{equation}
%Thus our upper bound (theorem \ref{upper}) holds for the weakest model of %randomized protocols.

The Gap Hamming Distance problem is the problem of computing the following partial function: 
$$
\mathsf{GHD}^n_{L, U}(x, y) = 
\begin{cases}
0 & H(x, y) \le L,\\
1 & H(x, y) \ge U,\\
\mathsf{undefined} & U < d(x, y) < L,
\end{cases}\qquad \mbox{for}\, x, y\in\{0, 1\}^n.
$$
%One may see that this definition is equivalent to the definition \ref{GHD}.
%\subsection{Deterministic and Non-deterministic Communication Complexity}

A protocol $\pi$ is called \emph{deterministic}, if $\pi$ does not use any randomness.

\begin{definition}
We say that a deterministic protocol computes $f$, if for every possible value  $i\in\{0, 1\}$ and for every pair of inputs from $f^{-1}(i)$ protocol outputs $i$.
Deterministic communication complexity of $f$ is defined as
$$
\D(f) = \min\limits_{\pi}CC(\pi),
$$
where minimum is over all deterministic protocols that compute $f$.

\end{definition}
\subsection{Codes}

In section \ref{section_4} we will use the notion of covering codes.
\begin{definition} A set $C\subset \{0, 1\}^n$ is called a \emph{covering code} of radius $r$, if
$$\forall x\in\{0, 1\}^n \qquad\exists y\in C\qquad H(x, y) \le r.$$
\end{definition}

Obviously, the size a covering code of radius $r$ is at least
$$\frac{2^n}{V_2(n, r)}.$$
There are covering codes with the almost optimal size.
\begin{proposition}[\cite{cohen1997covering}]
\label{good_covering_code}
There is a covering code in $\{0, 1\}^n$ of radius $r$ and size at most $O\left(\frac{n2^n}{V_2(n, r)}\right)$.
\end{proposition}
We will also use the fact that Hamming ball is the largest set among all subsets of $\{0, 1\}^n$ with the same diameter.  
\begin{definition}
Diameter of the set $A\subset \{0, 1\}^n$ is equal to
$$\diam(A) = \max\limits_{x, y\in A} H(x, y).$$
\end{definition}
\begin{theorem}[\cite{cohen1997covering}]
\label{size_theorem}
If $B\subset \{0, 1\}^n$, $\diam(B) \le 2r$ and $n \ge 2r + 1$, then $$|B| \le V_2(n, r).$$
\end{theorem}
%We define non-deterministic communication complexity in a combinatorial way.
%\begin{definition}
%Non-deterministic communication complexity of $f$, denoted by $\N^1(f)$, is equal to 
%$$\min\limits_R\lceil\log_2|R|\rceil,$$
%where minimum is over all $R\subset 2^{\mathcal{X}\times\mathcal{Y}}$ such that
%\begin{itemize}
%\item $\forall r\in R \qquad\exists A\subset \mathcal{X}, B\subset \mathcal{Y}\qquad r = A\times B$;
%\item $f^{-1}(1)\subset\bigcup\limits_{r\in R} r, \qquad \left(\bigcup\limits_{r\in R} r\right) \cap f^{-1}(0) = \varnothing$.
%\end{itemize}
%\end{definition}
%Co-non-deterministic communication complexity of $f$, denoted by $\N^0(f)$, can be defined in the similar way by replacing $1$ and $0$ in the definition.

%Non-deterministic communication complexity lower bounds deterministic communication complexity.
%\begin{proposition}
%\label{lower_bound_on_deterministic_complexity}
%$$\D(f) \ge \N^1(f), \qquad \D(f) \ge \N^0(f).$$
%\end{proposition}

\section{Upper Bound on One-Sided Error Communication Complexity of $\GHD$}
\label{section_3}

Consider any $x, y\in \mathbb{R}^t$. The scalar product and length of a vector are defined in the usual way
$$\langle x, y\rangle = \sum\limits_{i = 1}^t x_i y_i, \qquad \|x\| = \sqrt{\langle x, x\rangle}.$$
Let $\mathcal{US}(t)$ denote the uniform distribution on $(t - 1)$--dimensional unit sphere.
\begin{proposition}[\cite{marsaglia1972choosing}]
\label{us}
$\mathcal{US}(t)$ is equal to the distribution of the following vector
$$\frac{(Z_1, Z_2, \ldots, Z_t)}{\sqrt{Z_1^2 + \ldots + Z_t^2}},$$
where $Z_1, \ldots Z_t$ are independent random variables and for each of them we have that $Z_i\sim \mathcal{N}(0, 1)$.
\end{proposition}

\begin{lemma}
\label{uniform}
If $Z\sim \mathcal{US}(t)$, then for each $x\in\mathbb{R}^t$ we have
$$\E \langle x, Z\rangle^2 = \frac{\|x\|^2}{t}.$$
\end{lemma}
\begin{proof}
Let $Z_1, \ldots, Z_t$ be random variables from Proposition \ref{us}. Assume that $x = (1, 0, \ldots, 0)$. Then we have
$$\langle x, Z\rangle^2 = \frac{Z_1^2}{Z_1^2 + \ldots + Z_t^2}.$$
Random variables 
$$\frac{Z_1^2}{Z_1^2 + \ldots + Z_t^2},\, \frac{Z_2^2}{Z_1^2 + \ldots + Z_t^2},\, \ldots, \frac{Z_t^2}{Z_1^2 + \ldots + Z_t^2}$$
are identically distributed. Hence
$$1 = \E\frac{Z_1^2 + \ldots + Z_t^2}{Z_1^2 + \ldots + Z_t^2} = t\E\frac{Z_1^2}{Z_1^2 + \ldots + Z_t^2} = t\E (x, Z)^2.$$
Thus lemma is proved for $x = e_1 = (1, 0, \ldots, 0)$.

 Consider any other $x\in\mathbb{R}^t$. If $x = 0$, lemma is obvious. Otherwise there exists an orthogonal $n\times n$ matrix $A$ such that $\frac{x}{\|x\|} = Ae_1$. Now consider the vector $A^T Z$.
Proposition \ref{us} implies that vectors $A^T Z$ and $Z$ are identically distributed. Hence
$$\E \langle x, Z\rangle^2 = \|x\|^2 \E \langle Ae_1, Z\rangle^2 = \|x\|^2\E \langle e_1, A^T Z\rangle^2 = \|x\|^2 \E \langle e_1, Z\rangle^2 = \frac{\|x\|^2}{t}.$$
\end{proof}

Now we are able to construct the protocol for Theorem \ref{upper_bound}.

\begin{proof}[Proof of Theorem \ref{upper_bound}]
Set $b = \left\lceil 4n \sqrt[3]{\frac{s}{U}} \right\rceil$. If $b > n$, then the theorem \ref{upper_bound} states that $R^0_{e^{-s}}(\GHD^n_{U,\, L})$ is linear in $n$, which is trivial. Therefore we will assume that $b\le n$. Communication complexity of the protocol will be $O(b\log n)$. Set $a = \left\lceil \frac{n}{b}\right\rceil$ and
$$x_{n + 1} = \ldots = x_{ab} = y_{n + 1} = \ldots = y_{ab} = 0.$$

Note that

\begin{equation}
\label{ab}
ab = \left\lceil \frac{n}{b}\right\rceil b \ge \frac{n}{b}\cdot b = n, \qquad ab = \left\lceil\frac{n}{b}\right\rceil b\le \left(\frac{n}{b} + 1\right)b = n + b \le 2n.
\end{equation}

Alice and Bob transform their inputs $x, y$ to vectors $\alpha, \beta\in \mathbb{R}^{ab}$, where
$$\alpha = (x_1, \ldots, x_n, x_{n + 1}, \ldots, x_{ab}), \qquad \beta = (y_1, \ldots, y_n, y_{n + 1}, \ldots, y_{ab}).$$

Note that $H(x, y) = \|\alpha - \beta\|^2$. Alice and Bob divide $\alpha$ and $\beta$ into $b$ blocks of size $a$:
$$\alpha_i = (x_{ia - a + 1}, \ldots, x_{ia}),\qquad \beta_i = (y_{ia - a + 1}, \ldots, y_{ia}), \qquad i = 1, \ldots b.$$

The protocol runs as follows. Alice and Bob sample $b$ independent random vectors $U_1, \ldots, U_b$, each of them according to the distribution $\mathcal{US}(a)$. Then Alice computes $b$ numbers
$$\langle \alpha_1, U_1\rangle, \ldots, \langle \alpha_b, U_b\rangle,$$
and sends their approximations to Bob. More specifically, 
let $r_i$ be the closest to $\langle \alpha_i, U_i\rangle$ number in $\left\{\frac{m}{n^3}\left.\right| m\in\mathbb{Z}\right\}$. Note that

\begin{equation}
\label{approximation}
|r_i - \langle \alpha_i, U_i\rangle| \le \frac{1}{n^3}.
\end{equation} Alice sends $r_1, \ldots, r_b$ to Bob, each number specified by $O(\log n)$ bits. 
 Bob computes
$$T^\prime = (r_1 - \langle \beta_1, U_1\rangle)^2 + \ldots + (r_b - \langle \beta_b, U_b\rangle)^2,$$
If $T^\prime > L + \frac{5}{n}$, then Bob sends 1 to Alice. Otherwise Bob sends 0 to Alice. 

Communication complexity of the protocol is $O(b\log n)$. Now we have to estimate error probability. We first show that $T^\prime \le L + \frac{5}{n}$ whenever $H(x, y) \le L$ and thus the protocol does note err in this case. To this end consider the random variable 
$$T = \langle \alpha_1 - \beta_1, U_1\rangle^2 + \ldots + \langle \alpha_b - \beta_d, U_b\rangle^2.$$
Note that
$$
\begin{aligned}
H(x, y) &= \|\alpha - \beta\|^2 = \|\alpha_1 - \beta_1\|^2 + \ldots + \|\alpha_b - \beta_b\|^2\\
&\ge \langle \alpha_1 - \beta_1, U_1\rangle^2 + \ldots + \langle\alpha_b - \beta_b, U_b\rangle^2 = T.
\end{aligned}
$$

Let us show that $|T^\prime - T|$ is at most $\frac{5}{n}$. 
Denote $P_i =  r_i - \langle \beta_i, U_i\rangle$ and $Q_i = \langle \alpha_i - \beta_i,  U_i\rangle$. By definition

$$T^\prime = \sum\limits_{i = 1}^b P_i^2, \qquad T = \sum\limits_{i = 1}^b Q_i^2.$$
Thus $|T^\prime - T| \le \sum\limits_{i = 1}^b |P_i^2 - Q_i^2| =  \sum\limits_{i = 1}^b |P_i - Q_i| \cdot |P_i + Q_i|$.  Let us bound $|P_i - Q_i|$ and $|P_i + Q_i|$ separately. By \eqref{approximation} $|P_i - Q_i| \le \frac{1}{n^3}$.  By definition $|P_i + Q_i|$ is at most

\begin{align*}
|P_i + Q_i| &= \left| r_i + \langle\alpha_i, U_i\rangle - 2\langle \beta_i, U_i\rangle\right|\\
&\le |r_i| + \left|\langle \alpha_i, U_i\rangle\right| + 2\left|\langle \beta_i, U_i\rangle\right|\\
&\le 2 \left|\langle \alpha_i, U_i\rangle\right| + \frac{1}{n^3} + 2\left|\langle\beta_i, U_i\rangle\right|
\end{align*}
(again we use that $r_i$ is $\frac{1}{n^3}$-close to $\langle\alpha_i, U_i\rangle$). Coordinates of $\alpha_i$ and $\beta_i$ are zeros and ones  and there are at most $n$ ones among them. Hence

$$ \left|\langle \alpha, U_i\rangle\right| \le \|\alpha_i\|\le \sqrt{n}, \qquad \left|\langle \beta_i, U_i\rangle\right| \le \|\beta_i\|\le \sqrt{n},$$
   and therefore $|P_i + Q_i| \le 4\sqrt{n} + \frac{1}{n^3} \le 5n$. Finally
 \begin{align*}
 |T - T^\prime| \le \sum\limits_{i = 1}^b |P_i - Q_i| \cdot |P_i + Q_i| \le b \cdot\frac{1}{n^3} \cdot5n \le \frac{5}{n},
 \end{align*}
 since $b\le n$.

Assume that $H(x, y) = \|\alpha - \beta\|^2 \le L$. Then
$$T^\prime \le T + \frac{5}{n} \le L + \frac{5}{n}.$$
In this case the protocol always outputs 0. 

Now assume that $H(x, y) = \|\alpha - \beta\|^2 \ge U$. We will show that event $T^\prime \le L + \frac{5}{n}$ happens with small probability. By Lemma \ref{uniform} we have that
$$\E T = \frac{\|\alpha_1 - \beta_1\|^2}{a} + \ldots +  \frac{\|\alpha_b - \beta_b\|^2}{a} = \frac{\|\alpha - \beta\|^2}{a} = \frac{H(x, y)}{a} \ge \frac{U}{a}.$$

For each $i = 1, \ldots b$ we have
$$\langle \alpha_i - \beta_i, U_i\rangle^2 \in \left[0, \,\|\alpha_i - \beta_i\|^2\right]$$
with probability 1. To finish the proof we use the Hoeffding inequality:

\begin{proposition}[\cite{hoeffding1963probability}]
If random variables $\chi_1, \ldots, \chi_m$ are independent and for each $i = 1, \ldots m$
$$\chi_i\in[a_i, b_i]$$
with probability 1, then for every positive $\delta$
$$\Pr\left[\chi_1 + \ldots + \chi_m \le \E(\chi_1 + \ldots + \chi_m) - \delta\right] \le \exp\left\{ -\frac{2\delta^2}{\sum\limits_{i = 1}^m(b_i - a_i)^2}\right\}.$$
\end{proposition}

The following chain of inequalities finishes the proof.

\begin{align}
\label{chain_1}
\Pr\left[T^\prime \le L + \frac{5}{n}\right] &\le \Pr\left[T \le L + \frac{10}{n}\right]\\
\label{chain_2}
&\le \Pr\left[T \le \E T - E T/2\right]\\
\label{chain_3}
&\le \exp\left\{-\frac{(\E T)^2}{2(\|\alpha_1 - \beta_1\|^4 + \ldots + \|\alpha_b - \beta_b\|^4)}\right\}\\
\label{chain_4}
&\le \exp\left\{-\frac{\|\alpha - \beta\|^4/a^2}{2a(\|\alpha_1 - \beta_1\|^2 + \ldots + \|\alpha_b - \beta_b\|^2)}\right\}\\
\label{chain_5}
&\le \exp\left\{-\frac{U}{2a^3}\right\}\\
\label{chain_6}
 &\le \exp\{-s\}.
\end{align}

Let us explain it step by step.
\begin{itemize}
\item\eqref{chain_1} holds because $|T - T^\prime| \le \frac{5}{n}$;
\item First of all, by definition of $b$ and since $s \ge \frac{\left(L + \frac{10}{n}\right)^3}{U^2}$ we have

$$b^3 \ge \frac{4^3 sn^3}{U} \ge \frac{4^3 \frac{\left(L + \frac{10}{n}\right)^3}{U^2} n^3}{U} = \left(\frac{4n\left(L + \frac{10}{n}\right)}{U}\right)^3,$$
and hence $b \ge \frac{4n\left(L + \frac{10}{n}\right)}{U}$. Now recall that $\E T \ge \frac{U}{a}$ and by \eqref{ab} $ab \le 2n$. Therefore
$$\frac{ET}{2} \ge \frac{U}{2a} \ge \frac{bU}{4n} \ge L + \frac{10}{n}$$
and \eqref{chain_2} follows.

\item \eqref{chain_3} holds because of Hoeffding inequality, applied to 
$$T = \langle \alpha_1 - \beta_1, U_1\rangle^2 + \ldots + \langle \alpha_b - \beta_b, U_b\rangle^2.$$

\item \eqref{chain_4} holds because $\|\alpha_i- \beta_i\|^2 \le a$ and hence
$$\|\alpha_i- \beta_i\|^4 \le a\|\alpha_i- \beta_i\|^2.$$

\item \eqref{chain_5} holds because 
$$\|\alpha_1 - \beta_1\|^2 + \ldots + \|\alpha_b - \beta_b\|^2 = \|\alpha - \beta\|^2$$ and $\|\alpha - \beta\|^2\ge U$.

\item For \eqref{chain_6} again recall that $a \le \frac{2n}{b}$ and $b^3 \ge \frac{4^3 sn^3}{U}$.
$$\frac{U}{2a^3} \ge \frac{U}{2 \left(\frac{2n}{b}\right)^3} = \frac{b^3 U}{16n^3} \ge s.$$
\end{itemize}
\end{proof}

\section{Deterministic Communication Complexity of $\GHD^n_{0,\, t}$}
\label{section_4}

\subsection{Proof of Theorem \ref{det}}
Observe that
$$\D(\GHD^n_{0,\, n}) = 2 = n - \log_2 V_2\left(n, \left\lfloor \frac{n}{2}\right\rfloor\right) + O(\log n).$$
Hence we can assume that $t < n$.  

Consider the protocol $\pi$ witnessing $\D(\GHD^n_{0,\, t})$. Let $\pi(x, y)$ denote the leaf in protocol in which Alice and Bob come when Alice have $x$ on input and Bob has $y$ on input. If $l$ is a 0--leaf of the protocol $\pi$, consider the set
$$A_l = \left\{x\in\{0, 1\}^n\left.\right| \pi(x, x) = l\right\}$$
Note that $\diam(A_l) \le t - 1$. Indeed, assume that $x, y\in A_l$ and $H(x, y) \ge t$. At the same time
$$\pi(x, x) = l,\qquad \pi(y, y) = l\implies \pi(x, y) = l.$$
This contradicts the fact that $l$ is a 0--leaf of $\pi$. Observe that
$$t - 1 = 2\left(\frac{t}{2} - \frac{1}{2}\right)\le 2\left\lfloor\frac{t}{2}\right\rfloor, \qquad n \ge t + 1\ge 2\left\lfloor\frac{t}{2}\right\rfloor + 1.$$

Hence by Theorem \ref{size_theorem}
$$|A_l| \le V_2\left(n, \left\lfloor\frac{t}{2}\right\rfloor\right).$$

If both parties have the same $x\in\{0, 1\}^n$ on input, they must come to some 0--leaf. Hence
$$\left\{A_l\left.\right| l \,\,\mbox{is a 0--leaf of $\pi$ }\right\}$$
is a covering of $\{0, 1\}^n$. This covering has size at most $2^{CC(\pi)}$ and each set of the covering has size at most $V_2\left(n, \left\lfloor\frac{t}{2}\right\rfloor\right)$. Therefore
$$2^{CC(\pi)} V_2\left(n, \left\lfloor\frac{t}{2}\right\rfloor\right) \ge 2^n$$
and
$$\D(\GHD^n_{0,\, t}) = CC(\pi) \ge n - \log_2 V_2\left(n, \left\lfloor\frac{t}{2}\right\rfloor\right).$$

Let us prove the upper bound on $\D(\GHD^n_{0,\, t})$. 
Let $C$ be the covering code of radius $\left\lfloor\frac{t - 1}{2}\right\rfloor$ and size at most 
$$O\left(\frac{n2^n}{V_2\left(n, \left\lfloor\frac{t - 1}{2}\right\rfloor\right)}\right),$$
existing by Proposition \ref{good_covering_code}.

Alice computes
$$c = \arg\min\limits_{z\in C} H(z, x)$$
and sends $c$ to Bob. Since $c\in C$, it takes at most
$$\log_2 |C| + 1 = \log_2 O\left(\frac{n2^n}{V_2\left(n, \left\lfloor\frac{t - 1}{2}\right\rfloor\right)}\right) = n - \log_2 V_2\left(n, \left\lfloor\frac{t - 1}{2}\right\rfloor\right) + O(\log n)$$
bits. If $H(c, y) \le \left\lfloor\frac{t - 1}{2}\right\rfloor$, Bob sends 0 to Alice. Otherwise, Bob sends 1 to Alice.

Let us prove that the described protocol computes $\GHD^n_{0,\, t}$. Note that by definition of $c$ and $C$ we have $H(c, x) \le \left\lfloor\frac{t - 1}{2}\right\rfloor$. Hence if $x = y$, then $H(c, y)\le \left\lfloor\frac{t - 1}{2}\right\rfloor$. Assume now that $H(x, y) \ge t$. Then
$$H(x, c) + H(y, c) \ge H(x, y) \ge t > 2\left\lfloor\frac{t - 1}{2}\right\rfloor$$
and hence
$$H(y, c) >  2\left\lfloor\frac{t - 1}{2}\right\rfloor - H(x, c) \ge  \left\lfloor\frac{t - 1}{2}\right\rfloor.$$

Observe that
$$
\begin{aligned}
V_2\left(n, \left\lfloor\frac{t}{2}\right\rfloor\right) &\le V_2\left(n, \left\lfloor\frac{t - 1}{2}\right\rfloor\right) + \binom{n}{\left\lfloor \frac{t}{2}\right\rfloor} \\
&= V_2\left(n, \left\lfloor\frac{t - 1}{2}\right\rfloor\right) + \binom{n}{\left\lfloor \frac{t}{2}\right\rfloor - 1}\cdot \frac{n - \left\lfloor\frac{t}{2}\right\rfloor}{\left\lfloor\frac{t}{2}\right\rfloor}\\
&\le (1 + n) V_2\left(n, \left\lfloor\frac{t - 1}{2}\right\rfloor\right).
\end{aligned}
$$
Therefore the communication complexity of the protocol is at most
$$ n - \log_2 V_2\left(n, \left\lfloor\frac{t - 1}{2}\right\rfloor\right) + O(\log n) \le  n - \log_2 V_2\left(n, \left\lfloor\frac{t}{2}\right\rfloor\right) + O(\log n).$$

\subsection{Application to the Number of Distinct Elements (Proof of Theorem \ref{streaming_lower_bound})}

Let $F_0$ denote the number of distinct elements in a given data stream of size $2n$ with tokens drawn from the universe $U = \{1, 2, \ldots, 2n\}$. We say that 
a deterministic $p$-pass streaming algorithm $A$ with memory $S$ for computing $F_0$ is \emph{$c$--approximate} if $A$ outputs 
a number $E$ such that $F_0 \le E < cF_0$. We claim that for $c < 2$ A requires $\Omega\left(\frac{n(2 - c)^2}{p}\right)$ memory. Let us start with the case $p = 1$. 

The first result of that kind was proved in \cite{alon1996space}. It states that if $|E - F_0| < cF_0$, where $c=0.1$,
then $A$ requires $\Omega(n)$ memory. A linear  lower 
bound for memory for a larger $c$ can be obtained by reduction to the deterministic communication complexity of equality, as it done, for example, in \cite{chakrabarti2009data}. Indeed, for each $\alpha < \frac{1}{2}$ there is an error-correcting code $\mathsf{ECC}:\{0, 1\}^k\to\{0, 1\}^n$ with relative distance at least $\alpha$ and $k = \Omega_\alpha(n)$. 
Assume that Alice has $x\in\{0, 1\}^k$
and Bob has $y\in\{0, 1\}^k$ 
their inputs. They want to decide whether $x = y$. Alice and Bob transform their inputs into 2 data streams $u$ and $v$:
$$u = \langle u_1, u_2, \ldots, u_n\rangle, \qquad v = \langle v_1, v_2, \ldots, v_n\rangle,$$
where

\begin{equation}
\label{reduction}
u_i = n\cdot\mathsf{ECC}(x)_i  + i, \qquad v_i = n\cdot\mathsf{ECC}(y)_i  + i.
\end{equation}

Alice emulates $A$ on $u$. Then, using $S$ bits, she sends 
a description of the current state of $A$ to Bob and Bob 
emulates $A$ on $v$, starting with the state he received from Alice. 
Finally Bob knows the output of $A$ 
for the stream that is equal 
to  the concatenation of $u$ and $v$. Notice that the number of the distinct elements $F_0$ in this concatenation equals  $n + H\left(\mathsf{ECC}(x), \mathsf{ECC}(y)\right)$. If $A$ is a 
$(1 + \alpha)$--approximate (that is, $c = 1 + \alpha$), then Bob is able to decide whether $x = y$ or not. 
Indeed, if $x = y$, than $E < c F_0 = (1 + \alpha) n$. If $x \neq y$, then by definition of $\mathsf{ECC}$ we have that $E \ge F_0 \ge n + \alpha n = (1 + \alpha)n$. As deterministic 
communication complexity of the equality predicate
on $k$-bit strings is $k$, a linear lower bound $S = \Omega(k) = \Omega(n)$ 
for the space complexity of $A$ for $c < \frac{3}{2}$ follows.

In this argument we only needed
a linear lower bound for 1-round communication complexity of equality predicate,
which is trivial. However for arbitrary $p$ we already need 
a linear lower bound  for complexity of equality predicate
for $2p$-round protocols. The lower bound for the space 
complexity we obtain by this argument becomes  $\Omega(n/p)$, as 
in each round Alice in Bob exchange $S$ bits.

Instead of binary error--correcting codes one can use error--correcting codes 
with a 
larger alphabet and relative distance close to 1. The same reduction provides a linear lower bound for $c < 2$. The point is that the size of the universe increases and the problem becomes harder.

Theorem \ref{det} implies a linear lower bound on the space complexity of $A$ for $c < 2$ in the case when the size of the universe and the size of a data stream are equal. Indeed, assume that Alice has $x\in\{0, 1\}^n$ and Bob has $y\in\{0, 1\}^n$, as their inputs.
Assume also that they are  promised that 
either $x = y$ or $H(x, y)\ge t = \lceil n(c - 1)\rceil$. Again, Alice and Bob transform their input into data streams $u$ and $v$ but with \eqref{reduction} replaced by
$$u_i = nx_i  + i, \qquad v_i = n y_i  + i.$$
The expression for $F_0$ becomes $F_0 = n + H(x, y)$. Thus Alice and Bob can solve $\GHD^n_{0,\, t}$ using $2pS$ bits of communication. 
Indeed, if $x = y$, then $E < c F_0 = c n \le n + t$ since by definition $t \ge (c - 1) n$. If $H(x, y) \ge t$, then $E \ge F_0 \ge n + t$.

We conclude that by theorem \ref{det} that $pS$ must be at least
$$
pS = \Omega\left(n - \log_2 V_2\left(n, \frac{t}{2}\right) + \log n\right) = \Omega\left( \left(\frac{1}{2} - \frac{t}{2n}\right)^2 n\right) = \Omega(n (2 - c)^2 ).
$$

\bibliographystyle{acm}
\bibliography{ref}

\begin{thebibliography}{10}

\bibitem{alon1996space}
{\sc Alon, N., Matias, Y., and Szegedy, M.}
\newblock The space complexity of approximating the frequency moments.
\newblock In {\em Proceedings of the twenty-eighth annual ACM symposium on
  Theory of computing\/} (1996), ACM, pp.~20--29.

\bibitem{brody2010better}
{\sc Brody, J., Chakrabarti, A., Regev, O., Vidick, T., and De~Wolf, R.}
\newblock Better gap-hamming lower bounds via better round elimination.
\newblock In {\em Approximation, Randomization, and Combinatorial Optimization.
  Algorithms and Techniques}. Springer, 2010, pp.~476--489.

\bibitem{chakrabarti2009data}
{\sc Chakrabarti, A.}
\newblock Data stream algorithms.
\newblock {\em Computer Science 49\/}, 149.

\bibitem{chakrabarti2012optimal}
{\sc Chakrabarti, A., and Regev, O.}
\newblock An optimal lower bound on the communication complexity of
  gap-hamming-distance.
\newblock {\em SIAM Journal on Computing 41}, 5 (2012), 1299--1317.

\bibitem{cohen1997covering}
{\sc Cohen, G., Honkala, I., Litsyn, S., and Lobstein, A.}
\newblock {\em Covering codes}, vol.~54.
\newblock Elsevier, 1997.

\bibitem{hoeffding1963probability}
{\sc Hoeffding, W.}
\newblock Probability inequalities for sums of bounded random variables.
\newblock {\em Journal of the American statistical association 58}, 301 (1963),
  13--30.

\bibitem{indyk2003tight}
{\sc Indyk, P., and Woodruff, D.}
\newblock Tight lower bounds for the distinct elements problem.
\newblock In {\em Foundations of Computer Science, 2003. Proceedings. 44th
  Annual IEEE Symposium on\/} (2003), IEEE, pp.~283--288.

\bibitem{marsaglia1972choosing}
{\sc Marsaglia, G.}
\newblock Choosing a point from the surface of a sphere.
\newblock {\em The Annals of Mathematical Statistics 43}, 2 (1972), 645--646.

\bibitem{sherstov2012communication}
{\sc Sherstov, A.~A.}
\newblock The communication complexity of gap hamming distance.
\newblock {\em Theory of Computing 8}, 1 (2012), 197--208.

\bibitem{sudan2001algorithmic}
{\sc Sudan, M.}
\newblock {\em Algorithmic Introduction to Coding Theory: Lecture Notes}.
\newblock 2001.

\bibitem{vidick2012concentration}
{\sc Vidick, T.}
\newblock A concentration inequality for the overlap of a vector on a large
  set, with application to the communication complexity of the
  gap-hamming-distance problem.
\newblock {\em Chicago Journal of Theoretical Computer Science 1\/} (2012).

\end{thebibliography}

\end{document}